\documentclass[12pt]{article}
\usepackage{amsmath}
\usepackage{amsthm}
\usepackage{amsfonts}
\usepackage{graphicx}
\usepackage{latexsym}
\usepackage{amssymb}
\usepackage{color}
\usepackage{url}

\usepackage{slashbox}
\usepackage{xfrac}

\DeclareMathAlphabet{\mathbfsl}{OT1}{ppl}{b}{it} 

\makeatletter
\DeclareRobustCommand{\nsbinom}{\genfrac[]\z@{}}
\makeatother

\newcommand{\abs}[1]{\left|#1\right|}

\newcommand{\field}[1]{\mathbb{#1}}

\newcommand{\Z}{\field{Z}}

\newcommand{\F}{\field{F}}

\newcommand{\cA}{{\cal A}}
\newcommand{\cB}{{\cal B}}
\newcommand{\cC}{{\cal C}}
\newcommand{\cD}{{\cal D}}
\newcommand{\cF}{{\cal F}}

\newcommand{\cS}{{\cal S}}
\newcommand{\cT}{{\cal T}}
\newcommand{\cP}{{\cal P}}

\newcommand{\cV}{{\cal V}}

\newcommand{\cM}{{\cal M}}
\newcommand{\cN}{{\cal N}}

\newcommand{\linadd}{\kern1pt\mbox{\small$\boxplus$}\kern1pt}

\newcommand{\supp}[1]{ \text{supp} ({#1}) }
\newcommand{\wt}[1]{ \text{wt} ({#1}) }

\newtheorem{definition}{Definition}
\newtheorem{theorem}{Theorem}
\newtheorem{lemma}{Lemma}

\newtheorem{corollary}{Corollary}
\newtheorem{conjecture}{Conjecture}
\newtheorem{example}{Example}

\theoremstyle{definition}

\begin{document}

\bibliographystyle{plain}

\title{Non-Binary Diameter Perfect Constant-Weight Codes}

\author{
{\sc Tuvi Etzion}\thanks{Department of Computer Science, Technion,
Haifa 3200003, Israel, e-mail: {\tt etzion@cs.technion.ac.il}.
}
}

\maketitle

\begin{abstract}
Diameter perfect codes form a natural generalization for perfect codes. They are based on the
code-anticode bound which generalizes the sphere-packing bound. The code-anticode bound was proved by Delsarte for distance-regular
graphs and it holds for some other metrics too. In this paper we prove the bound for non-binary constant-weight codes
with the Hamming metric and characterize the diameter perfect codes and the maximum size anticodes for these codes.
We distinguish between six families of non-binary diameter constant-weight codes and four families of maximum size non-binary
constant-weight anticodes. Each one of these families of diameter perfect codes raises some different questions.
We consider some of these questions and leave lot of ground for further research. Finally, as a consequence,
some $t$-intersecting families related to the well-known Erd\"{o}s-Ko-Rado theorem, are constructed.
\end{abstract}

\vspace{0.5cm}


\vspace{0.5cm}



\newpage
\section{Introduction}

Perfect codes are one of the most fascinating structures in coding theory.
These codes meet the well-known sphere-packing bound. When the size of a ball does not depend on its center,
these codes are the largest ones for their length and minimum distance.
They are defined for variety of metrics in coding theory. For most such metrics other bounds which
improve on the sphere-packing bound, were developed. One of these bounds is the {\bf \emph{code-anticode bound}} which
was first introduced by Delsarte in his seminal work~\cite{Del73}. This bound is a direct generalization of
the sphere-packing bound. For this bound the concept of anticode is required.
An {\bf \emph{anticode}} $\cA$ with {\bf \emph{diameter}} $D$ in a space~$\cV$ is a subset of $\cV$, with a metric,
such that the maximum distance between elements in $\cA$ is $D$.
%
Delsarte~\cite{Del73} proved the code-anticode bound for distance-regular graphs. The bound was re-introduced and generalized by
Ahlswede, Aydinian, and Khachatrian~\cite{AAK01} who called any code which meets this bound, a $D$-diameter perfect code.
To apply the code-anticode bound, one must find the size of the maximum size anticode for the related space, metric, and parameters.
This diametric problem in some schemes, such as the Hamming scheme and the Johnson scheme was discussed in~\cite{AAK01}
and shown to be closely related to intersection problems for systems of finite sets which go back to Erd\"{o}s-Ko-Rado theorem
(see~\cite{AhKh97,AhKh98,EKR61} and references therein).

The proof of the code-anticode bound which was introduced in~\cite{AAK01} required a space and metric associated
with a graph which admits a transitive group of automorphisms. The proof in~\cite{AAK01} was demonstrated
on the Johnson graph. Nevertheless, no general proof was given for the theorem for the family of graphs
which admit a transitive group of automorphisms.
They have considered the related codes and anticodes in the Hamming scheme, the Johnson scheme, and the Grassmann scheme.
Other metrics, which are not associated with distance-regular graphs, for which the proof was generalized, are
the $\ell_\infty$ metric on the set of permutations~$S_n$~\cite{TaSc10},
the Lee metric in $\Z_m^n$~\cite{Etz11} and the Kendall~$\tau$-metric on the set of permutations~$S_n$~\cite{BuEt15}.
The bound was also considered for the rank-distance on Ferrers diagrams in~\cite{EGRW16}.
For each one of these spaces and metrics some diameter perfect codes were introduced in these papers.

In this paper we further expand our knowledge on diameter perfect codes and optimal anticodes. We consider the set of
all words of length $n$ with weight $w$ over an alphabet $Q$ with $q$ symbols, $q >2$, where the metric used is the Hamming distance.
The graph (whose vertices are the space) associated with these words and metric will be denoted by J$_q(n,w)$. Its set of vertices
consists of all the words of length $n$ and weight $w$ over an alphabet $Q$ of size $q$, say $\Z_q$. Two vertices are connected by an undirected
edge if the Hamming distance of their associated words is one. Hence, when $q>2$ the related graph is not connected unless $w=n$.
Therefore, the graph is not distance-regular, unless $w=n$, which implies that the direct arguments in~\cite{AAK01,Del73} do not apply.
It should be noted that the graph is a union of disjoint distance-regular graphs, each one is formed from codewords which have the same support.
This graph was considered before when $q=3$ and $w=n-1$ in~\cite{Kro08,KOP16} who provided some arguments for the
code-anticode bound for this family of codes. Nevertheless, similarly to the proof in~\cite{AAK01},
the proof in~\cite{Kro08,KOP16} is not complete. In~\cite{Kro08} it is proved that there is only one such diameter perfect code
with minimum distance 4, based on the nonexistence of related perfect colorings~\cite{Fon07}. There is also a construction
in~\cite{Kro08} of diameter perfect codes with minimum distance 5. In~\cite{KOP16} there is a proof that a
diameter perfect code with minimum distance 5 and length 16 does not exists. A diameter perfect code with minimum distance 5 and length 64
is derived from a construction in~\cite{BDMW}.
In this paper, we provide a general proof of the code-anticode bound for J$_q(n,w)$, $q >2$.
We distinguish between six families of diameter perfect codes
in J$_q(n,w)$ and four families of maximum size anticodes associated with them. It should be noted that
each family of non-binary diameter perfect constant-weight codes is a family of maximum size constant-weight codes
with the related parameters (the length of the codewords will be denoted by $n$, their weight by $w$, the minimum distance of the code by $d$, and
the alphabet size by $q$). This problem of finding the maximum size of such codes was extensively studied, e.g.~\cite{CDLL08,ChLi07,ZhGe13}
and references therein.

The rest of this paper is organized as follows. In Section~\ref{sec:code-anticode}, the sphere-packing bound, the
code-anticode bound, and its generalization presented in~\cite{AAK01}, are introduced. The
local inequality lemma from which the code-anticode bound can be proved, is also presented in this section. Based on the code-anticode bound
the concept of a diameter perfect code, which is a generalization of a perfect code, is defined. Finally, the set of words
with weight $w$ over an alphabet with $q$ letters is defined and some preliminaries are given. In Section~\ref{sec:diam_John},
the code-anticode bound is proved for the Johnson scheme and it is shown that Steiner systems and complements of Steiner systems
are diameter perfect codes in this scheme. In Section~\ref{sec:no_bin_diam_John}, non-binary diameter perfect constant-weight codes
are discussed. The code-anticode bound is proved for non-binary constant-weight words and six families of diameter
perfect codes are characterized. In Section~\ref{sec:w=n}, non-binary diameter perfect constant-weight codes for which $w=n$
are discussed. These codes are derived from diameter perfect codes in the Hamming scheme. There is a one-to-one correspondence between
these two families of codes. In Section~\ref{sec:perfect}, non-binary perfect constant-weight codes for which the alphabet size is $2^k+1$,
the length is $n$, and the weight is $w=n-1$ are considered. Most of these codes are over ternary alphabet.
In Section~\ref{sec:GST}, we define the family of generalized Steiner systems and show that each
generalized Steiner system is a diameter perfect constant-weight code. In Section~\ref{sec:MDS-CW}, the family of MDS constant-weight codes
is defined. This family got its name from MDS codes since the minimum weight codewords of any MDS code form
such a code. Section~\ref{sec:d=w+1} is devoted to codes for which $d=w+1$. Codes in this family have exactly one codeword
for each support which consists of any $w$-subset of the $n$ coordinates. In Section~\ref{sec:MOA-CW}, the last family called
multiple orthogonal arrays constant-weight codes is defined. For this family $d < w$ and each $w$ coordinates are the supports
of codewords which form an orthogonal array. In Section~\ref{sec:anticodes}, we compare between the different maximum size anticodes
associated with the six families of diameter perfect constant-weight codes. We distinguish between four families of such anticodes.
Finally, in Section~\ref{sec:conclude} we summarize our results and suggest a list of open problems for future research.

\section{The Code-Anticode Bound and Preliminaries}
\label{sec:code-anticode}

Let $\cV$ be a space and $d:\cV \times \cV \longrightarrow \Z$ be a metric defined on $\cV$. In the metric $d$,
the {\bf \emph{ball}}, $\cB_e(x)$ of radius $e$ around $x \in \cV$ is the set of elements which are at distance $e$ from $x$,
i.e. $\cB_e (x) \triangleq \{ y ~:~ d(x,y) \leq e \}$. The metric is called {\bf \emph{regular}} if the size of the ball does not depend
on its center. In this case the ball of radius $e$ is denoted by $\cB(e)$.
The well-known sphere-packing bound is introduced in the following theorem.

\begin{theorem}
\label{thm:sphere_pack}
If $\cC \subseteq \cV$ is a code with minimum distance $2e+1$, associated with a regular metric, then
$$
\abs{\cC} \cdot \abs{\cB(e)} \leq \abs{\cV}~.
$$
\end{theorem}

A code which meets the bound of Theorem~\ref{thm:sphere_pack} is called a \emph{perfect code}.
Theorem~\ref{thm:sphere_pack} is generalized by the code-anticode bound which will be introduced now.
The bound is an immediate consequence of the following theorem which was proved by Delsarte~\cite{Del73}.

\begin{theorem}
\label{thm:ditinct_dist_bound}
Let $\cV$ be the set of vertices in a distance-regular graph and let $X$ and $Y$ be two set of vertices in $\cV$.
If the nonzero distances occurring between vertices of $X$ do not occur between vertices of $Y$, then
$$
\abs{X} \cdot \abs{Y} \leq \abs{\cV}~.
$$
\end{theorem}
\begin{corollary}
\label{cor:ditinct_dist_bound}
Let $\cV$ be the set of vertices in a distance-regular graph and let $X$ and $Y$ be two sets of vertices in $\cV$.
If the the minimum distance in $X$ is $D+1$ and the maximum distance in $Y$ is $D$, then
$$
\abs{X} \cdot \abs{Y} \leq \abs{\cV}~.
$$
\end{corollary}

Theorem~\ref{thm:ditinct_dist_bound} was proved in~\cite{Del73} for distance-regular graphs using the duality theorem
for linear programming. Much simpler proofs to the same bound for distance-regular graphs were given by~\cite{Del76}
and~\cite{Roo82}. Theorem~\ref{thm:ditinct_dist_bound}
was generalized by Ahlswede, Aydinian, and Khachatrian~\cite{AAK01} with the following lemma which is the key result
required for the definition of diameter perfect codes. It has several proofs, depending on the
space and the metric being considered. Nevertheless, there are spaces and metrics for which this lemma is not satisfied
and hence it cannot be used for those spaces and metrics.

\begin{lemma}
\label{lem:local_inequality}
Let $\cC_\cD$ be a code in a space $\cV$ with a metric $d: \cV \times \cV \longrightarrow \Z$ (satisfying some conditions),
where the distances between the codewords in
$\cC$ are taken from a subset $\cD$. Let~$\cA$ be a subset of $\cV$ and let $\cC'_\cD \subseteq \cA$ be the
largest code in $\cA$ with distances taken from $\cD$. Then
\begin{equation}
\label{eq:local_inequality}
\frac{\abs{\cC_\cD}}{\abs{\cV}} \leq \frac{\abs{\cC'_\cD}}{\abs{\cA}} ~.
\end{equation}
\end{lemma}

The conditions required for Lemma~\ref{lem:local_inequality} can be different depending on the space and the metric.
The lemma was proved for several such spaces and metrics, e.g.~\cite{BuEt15,Etz11,EGRW16,TaSc10}.
Lemma~\ref{lem:local_inequality} will be referred to as the {\bf \emph{local inequality lemma}} and~(\ref{eq:local_inequality})
as the {\bf \emph{local inequality bound}}. The local inequality lemma implies the code-anticode bound as follows.
Choose $\cA$ in Lemma~\ref{lem:local_inequality} to be an anticode with diameter $D$ and choose $\cC_\cD$ to be a code
with distances between codewords taken from the subset $\cD \triangleq \{ i ~:~ D+1 \leq i \leq \Delta \}$, for some $\Delta \geq D+1$.
Hence, $\cC'_\cD$ is a code with exactly one codeword and therefore~(\ref{eq:local_inequality}) implies that
\begin{equation}
\label{eq:code_anticode_bound}
\abs{\cC_\cD} \cdot \abs{\cA} \leq  \abs{\cV} ~,
\end{equation}
which will be referred to as the {\bf \emph{code-anticode bound}}.

Lemma~\ref{lem:local_inequality} was proved in~\cite{AAK01} only for the Johnson graph J$(n,w)$, but it was claimed that
it holds for any graph which admits a transitive group of automorphism. The proof for the Johnson graph was not trivial and a different proof,
based on the transitive operation,
should be adapted to different distance-regular graphs. In particular it can be easily adapted to the Hamming graph.
Distance-regular graphs are associated with association schemes~\cite{Del73} from which the most important one in coding theory
is the Hamming scheme associated with the Hamming distance. The following lemma is well-known in the Hamming scheme
as well as for any other metric. It can be proved using the triangle inequality.

\begin{lemma}
\label{lem:ball=anti}
A ball with radius $e$ in a metric is an anticode with diameter $2e$.
\end{lemma}
\begin{proof}
Let $y,z$ be two distinct words in a the ball centered at some word~$x$, i.e., $d(y,x) \leq e$ and $d(x,z) \leq e$.
Hence, by the triangle inequality,
$$
d(y,z) \leq d(y,x) + d(x,z) \leq 2e ,
$$
which implies that an ball with radius $e$ is an anticode with diameter $2e$.
\end{proof}

A code $\cC$ which attains~(\ref{eq:code_anticode_bound}) with equality is called a {\bf \emph{$D$-diameter perfect code}}.
The following theorem implies a tight connection between perfect codes and diameter perfect codes. It is an immediate consequence
from Theorem~\ref{thm:sphere_pack} and Lemma~\ref{lem:ball=anti}.

\begin{theorem}
\label{thm:perfect=diameter}
If $d:\cV \times \cV \longrightarrow \Z$ is a regular metric for which
the local inequality bound~(\ref{eq:local_inequality}) is satisfied,
then an $e$-perfect code in $\cV$ is a $(2e)$-diameter perfect code.
\end{theorem}

We continue with the spaces and metrics which will be discussed in this paper.

\begin{definition}
Let J$(n,w)$ denote the set of all $w$-subset of an $n$-set (or equivalently, the set of all the binary words of length $n$
and weight $w$), with the Johnson distance (defined as half of the Hamming distance).
\end{definition}

\begin{definition}
Let J$_q(n,w)$, $q > 2$, denote the set of all the words of length $n$ and weight $w$, over $\Z_q$, with the Hamming distance.
\end{definition}

Clearly, both J$(n,w)$ and J$_q(n,w)$ are contained in the Hamming scheme. But, J$(n,w)$ is a well-known association scheme,
while J$_q(n,w)$ does not define an association scheme and its related graph is not connected. But, since both subsets are part
of the Hamming scheme, they share the following property of the Hamming scheme implied by Theorem~\ref{thm:perfect=diameter}.

\begin{theorem}
\label{thm:eP_diamP}
An $e$-perfect code in J$(n,w)$ is a $(2e)$-diameter perfect code in J$(n,w)$. If the local inequality lemma holds for J$_q(n,w)$, then
an $e$-perfect code in J$_q(n,w)$ is a $(2e)$-diameter perfect code in J$_q(n,w)$.
\end{theorem}

The following definitions will be used in the sequel.
For a word $x=(x_1,x_2,\ldots,x_n)$ over an alphabet $Q$, the {\bf \emph{support}} of $x$, $\supp{x}$,
is the subset of nonzero coordinates in $x$, i.e., $\supp{x} \triangleq \{ i ~:~ x_i \neq 0,~ 1 \leq i \leq n \}$.
The {\bf \emph{weight}} of $x$, $\wt{x}$, is the number of nonzero coordinates in $x$,
i.e., $\wt{x} \triangleq \abs{ \{i ~:~ x_i \neq 0,~ 1 \leq i \leq n \} } = \abs{\supp{x}}$.
Let $Q$ denote an alphabet of size $q$ which contains a \emph{zero} and let $Q^* \triangleq Q \setminus \{ 0 \}$.

\begin{definition}
An $(n,d,w)_q$ code is a code of length $n$ over an alphabet with $q>2$ symbols having minimum Hamming distance $d$.
Let $A_q(n,d,w)$ denote the maximum size of an $(n,d,w)_q$ code.
\end{definition}

In the sequel, for some constructions and bounds we will require to use MDS codes and their nonlinear version of orthogonal arrays.
An {\bf \emph{orthogonal array}} OA$_\lambda (t,n,q)$ is an $(\lambda q^t) \times n$ array $\cM$ over an alphabet $Q$ with $q$ symbols
such that in each projection of $t$ columns from $\cM$ each ordered $t$-tuple of $Q$ appears in exactly $\lambda$ rows. When $\lambda =1$
the orthogonal array is denoted by OA$(t,n,q)$ and is called an {\bf \emph{orthogonal array with index unity}}. If such an orthogonal array with
index unity is a linear code, then the code is called a {\bf \emph{maximum distance separable code}}
(an MDS code in short). This code has length~$n$, dimension $t$, and minimum Hamming distance $n-t+1$.
An orthogonal array OA$(t,n,q)$ is also a $(n-t)$-diameter perfect code~\cite{AAK01} in the Hamming scheme.
The related maximum size anticode of length $n$, diameter $n-t$, and $q^{n-t}$ codewords is defined by
$$
\{ (a_1,a_2,\ldots,a_{n-t} \overbrace{0\cdots \cdots 0}^{t ~ \text{times}}) )   ~:~  a_i \in Q,~ 1 \leq i \leq n-t  \}~.
$$

There are some well-known bounds on the tradeoff between $t$, $n$, and $q$ (see for example~\cite[pp. 11-16]{Rag71}).

\begin{theorem}
\label{thm:b_OA1}
If there exists an OA$(2,n,q)$ then $n \leq q +1$.
\end{theorem}
\begin{theorem}
\label{thm:b_OA2}
Assume that there exists an OA$(t,n,q)$, where $t \geq 3$ and $q \geq t$.
\begin{enumerate}
\item If $q$ is even then $n \leq q+t-1$.

\item If $q$ is odd then $n \leq q+t-2$.
\end{enumerate}
\end{theorem}
\begin{theorem}
\label{thm:b_OA3}
If there exists an OA$(t,n,q)$, where $q \leq t$, then $n \leq t+1$.
\end{theorem}

The following result~\cite[pp. 317--331]{McSl77} is derived from MDS codes.
\begin{theorem}
\label{thm:MDS_conjecture}
If $t \leq n-2$, then there exists an OA$(t,n,q)$ if $n \leq q+1$ for each prime power~$q$
and $2 \leq t \leq q-1$. The only exception is when $q$ is even and $t \in \{3,q-1\}$ in which case $n \leq q+2$.
\end{theorem}

\section{Binary Diameter Perfect Constant-Weight Codes}
\label{sec:diam_John}

Our first step is to prove the local inequality lemma and as a consequence to prove the code-anticode bound.
The local inequality lemma for the Johnson scheme was proved in~\cite{AAK01} and for completeness the proof is presented here.
It will be interesting to look on the difference between the proof of this lemma for the Johnson scheme and
its proof for constant-weight codes over a non-binary alphabet (see Lemma~\ref{lem:gen_anticode_preJQ}).

\begin{lemma}
\label{lem:local_John}
Let $\cC_\cD$ be a code in J$(n,w)$ with distances between the codewords of $\cC_\cD$ are taken from a subset $\cD$.
Let $\cA$ be a subset of J$(n,w)$ and let $\cC'_\cD \subseteq \cA$ be the largest code in $\cA$ with distances taken from $\cD$. Then
\begin{equation}
\label{eq:local_John}
\frac{|\cC_\cD|}{\binom{n}{w}} \leq \frac{|\cC'_\cD|}{\abs{\cA}} ~.
\end{equation}
\end{lemma}
\begin{proof}
Consider the set of pairs
$$
\cP= \{(c,\pi) ~:~ c \in \cC_\cD, ~ \pi \in S_n,~ \pi (c) \in \cA \}.
$$
For a fixed $c \in \cC_\cD$ and a fixed $a \in \cA$ there are exactly $w! (n-w)!$ choices for a permutation~$\pi$, such that $a=\pi (c)$.
Hence, the number of pairs in $\cP$ equals to $\abs{\cC_\cD} \cdot \abs{\cA} \cdot w! \cdot (n-w)!$.

Note, that for each permutation $\pi$ and two elements $x,y \in \text{J}(n,w)$, we have that $d(\pi(x),\pi(y)) = d(x,y)$.
It implies that a fixed permutation $\pi \in S_n$ can transfer the elements of $\cC_\cD$ into at most $\abs{\cC'_\cD}$ elements
of $\cA$. Therefore, each  permutation $\pi$ contributes at most $\abs{\cC'_\cD}$ pairs to $\cP$,
and hence the number of pairs in $\cP$ is at most $\abs{\cC'_\cD} \cdot n!$ which implies that
$$
\abs{\cC_\cD} \cdot \abs{\cA} \cdot w! \cdot (n-w)! \leq \abs{\cC'_\cD} \cdot n!~,
$$
and the claim of the lemma follows.
\end{proof}
Lemma~\ref{lem:local_John} implies that the code-anticode bound is satisfied for the Johnson scheme.

\begin{corollary}
If $\cC$ is a code in J$(n,w)$ with minimum Johnson distance $D+1$ and $\cA$ is an anticode in J$(n,w)$ with maximum Johnson distance $D$, then
$$
\abs{\cC} \cdot \abs{\cA} \leq \binom{n}{w}~.
$$
\end{corollary}

The following lemma is readily verified.
\begin{lemma}
\label{lem:anticode_n_w_t}
The set $\cA(n,w,t)$, where $0 \leq t \leq w \leq \frac{n}{2}$, defined by
$$
\cA(n,w,t) \triangleq \{ ( \overbrace{1\cdots \cdots 1}^{t ~ \text{times}} ,a_1,\ldots,a_{n-t}) ~:~ a_j \in \{0,1\},~1 \leq j \leq n-t,~ \wt{a_1,\ldots,a_{n-t}}=w-t  \},
$$
is an anticode in J$(n,w)$ whose diameter is $w-t$ and its size is~$\binom{n-t}{w-t}$.
\end{lemma}

For a binary code $\cC$ of length $n$, the {\bf \emph{complement}} of $\cC$, $\bar{\cC}$, is defined by
$$
\bar{\cC} \triangleq \{ (\bar{x}_1 , \bar{x}_2,\ldots,\bar{x}_n) ~:~ (x_1,x_2,\ldots,x_n) \in \cC \},
$$
where $\bar{b}$ is the binary complement of $b \in \{ 0,1 \}$.

\begin{lemma}
\label{lem:anticode_complementJ}
The set $\bar{\cA} (n,w,t)$, where $0 \leq t \leq w \leq \frac{n}{2}$, defined by
$$
\bar{\cA} (n,w,t) \triangleq \{ (\overbrace{0\cdots \cdots 0}^{t ~ \text{times}} ,a_1,\ldots,a_{n-t}) ~:~ a_j \in \{0,1\},~ \wt{a_1,\ldots,a_{n-t}}=n-w  \},
$$
is an anticode in J$(n,n-w)$ whose diameter is $w-t$ and its size is~$\binom{n-t}{w-t}$.
\end{lemma}
\begin{proof}
Clearly, $\bar{\cA} (n,w,t)$ is the complement of $\cA (n,w,t)$ and hence by Lemma~\ref{lem:anticode_n_w_t}, we have that
$\abs{\bar{\cA} (n,w,t)} = \binom{n-t}{w-t}$. Moreover, for each two words $x,y \in \text{J}(n,w)$, $d(x,y)=d(\bar{x},\bar{y})$
and hence the diameter of $\bar{\cA} (n,w,t)$ equals to the diameter of $\cA (n,w,t)$ which is $w-t$ by Lemma~\ref{lem:anticode_n_w_t}.
Thus, the claim of the lemma follows directly from Lemma~\ref{lem:anticode_n_w_t}.
\end{proof}

\begin{definition}
A Steiner system S$(t,w,n)$ is a pair $S=(\cN,B)$, where $\cN$ is an $n$-set and $B$ is a set of $w$-subsets (called {\bf \emph{blocks}})
from $\cN$, where each $t$-subset of $\cN$ is contained in exactly one block of $B$.
\end{definition}

Subsets can be translated to words, and vice versa, via the following definition.
The {\bf \emph{characteristic  vector}} of a $w$-subset $S$ of an $n$-set $\cN$ is a binary word of length $n$
and weight~$w$ whose $i$-th coordinate is a \emph{one} if and only if the $i$-th element of~$\cN$ is contained in~$S$.
In some cases, we are going to consider mixed language of words and subsets and the translation between the
two should be understood from the context.

The following well-known theorem is readily verified.
\begin{theorem}
\label{thm:ST_b_s}
The number of blocks in a Steiner system S$(t,w,n)$ is
$$
\binom{n}{t} {\Large \text{/}} \binom{w}{t}
$$
and its minimum Johnson distance is $w-t+1$.
\end{theorem}

\begin{theorem}
\label{thm:anticode_Steiner}
Any Steiner system S$(t,w,n)$ forms a $(w-t)$-diameter perfect code.
\end{theorem}
\begin{proof}
If $\cC$ is the code constructed from a Steiner system \text{S}$(t,w,n)$, then by Theorem~\ref{thm:ST_b_s}
its Johnson distance is $w-t+1$, and
$$
\abs{\cC} = \frac{\binom{n}{t}}{\binom{w}{t}} = \frac{\binom{n}{w}}{\binom{n-t}{w-t}} ~.
$$
On the other hand by the code-anticode bound $\abs{\cC} \cdot \abs{\cA} \leq \binom{n}{w}$, where $\cA$ is an anticode in J$(n,w)$
whose diameter is $w-t$, and therefore $\abs{\cA} \leq \binom{n-t}{w-t}$.

Since by Lemma~\ref{lem:anticode_n_w_t} the set $\cA (n,w,t)$ is an anticode in J$(n,w)$ of size $\binom{n-t}{w-t}$ and whose
diameter is $w-t$, the claim of the theorem follows.
\end{proof}
Theorem~\ref{thm:anticode_Steiner} was proved in~\cite{AAK01}, but it should be noted (as it was not mentioned in~\cite{AAK01}),
that except for the family of Steiner systems, there is another family of diameter perfect codes in J$(n,w)$.

\begin{theorem}
\label{thm:anticode_comp_Steiner}
The complement of a Steiner system S$(t,w,n)$ forms a $(w-t)$-diameter perfect code.
\end{theorem}
\begin{proof}
For each two words $x,y \in \text{J}(n,w)$, $d(x,y)=d(\bar{x},\bar{y})$ and hence the  complement of a Steiner system S$(t,w,n)$
has minimum Johnson distance $w-t+1$. By Lemma~\ref{lem:anticode_complementJ}, $\bar{\cA} (n,w,t)$ and $\cA (n,w,t)$ have the same
diameter $w-t$ and the same size $\binom{n-t}{w-t}$. Moreover, the weight of the codewords in $\bar{\cA} (n,w,t)$ is $n-w$ and
since also $\binom{n}{n-w}=\binom{n}{w}$, this implies, as in the proof of Theorem~\ref{thm:anticode_Steiner}, that the
complement of a Steiner system S$(t,w,n)$ forms a $(w-t)$-diameter perfect code in J$(n,n-w)$.
\end{proof}

\begin{corollary}
Any Steiner system S$(t,w,n)$ and any complement of a Steiner system S$(t,w,n)$ forms a $(w-t)$-diameter perfect code.
\end{corollary}

There is a well-known conjecture~\cite{Del73} that there are no nontrivial perfect codes in J$(n,w)$.
Steiner systems are embedded in diameter perfect codes in J$(n,w)$~\cite{AAK01} similar to Steiner systems embedded in perfect
codes in J$(n,w)$ as was proved in~\cite{Etz96}. This makes it very tempting to have a similar conjecture for diameter perfect codes.
\begin{conjecture}
\label{conj:Dbinary}
There are no nontrivial diameter perfect codes in J$(n,w)$, except for Steiner systems and their complements.
\end{conjecture}

The following well-known results, to which we will refer later, can be easily verified.
\begin{lemma}
\label{lem:derived}
If there exists a Steiner system S$(t,w,n)$, $t > 1$, then there exists a Steiner system S$(t-1,w-1,n-1)$.
\end{lemma}
\begin{corollary}
\label{cor:necesary_ST}
A necessary condition that a Steiner system S$(t,w,n)$ exists is that all the numbers $\frac{\binom{n-i}{t-i}}{\binom{k-i}{t-i}}$,
$0 \leq i \leq t-1$, are integers.
\end{corollary}

Steiner system were investigated throughout the years and a short survey by Colbourn and Mathon can be found in~\cite[pp. 102--110]{CoDi07}.
It was proved in~\cite{Kee14} and later in~\cite{GKLO16} that for each pair $(t,w)$,
there exists an $n_0$ such that for each $n \geq n_0$ the necessary conditions of Corollary~\ref{cor:necesary_ST} are also sufficient.
Unfortunately, the proof is nonconstructive and this $n_0$ is beyond our imagination.

\section{Non-Binary Diameter Perfect Constant-Weight Codes}
\label{sec:no_bin_diam_John}

The main part of our paper is devoted to non-binary diameter perfect constant-weight codes.
We distinguish between six families of such codes in J$_q(n,w)$, where $q>2$.

\begin{itemize}
\item[{\bf [F1]}] Non-binary diameter perfect constant-weight codes for which $w=n$.

\item[{\bf [F2]}] Diameter perfect constant-weight codes over an alphabet of size $2^k+1$  for which~${w=n-1}$ .

\item[{\bf [F3]}] Non-binary diameter perfect constant-weight codes which are generalized Steiner systems.

\item[{\bf [F4]}] Non-binary diameter perfect constant-weight codes for which $d=w$. These codes are
called maximum distance separable constant-weight codes. Each such code has ${\binom{n}{w} (q-1)}$ codewords.

\item[{\bf [F5]}] Non-binary diameter perfect constant-weight codes for which $d=w+1$. Such a code has $\binom{n}{w}$ codewords.

\item[{\bf [F6]}] Non-binary diameter perfect constant-weight codes for which $d<w$. These codes are
called multiple orthogonal arrays constant-weight codes. Each such code has ${\binom{n}{w} (q-1)^{w-d+1}}$ codewords.
\end{itemize}

Before we start our discussion on these six families of codes we have to prove the local inequality lemma for J$_q(n,w)$
with the Hamming metric. For the proof of this lemma it is required to prove the following simple lemma.

\begin{lemma}
\label{lem:lambda_for_all}
For each $q \geq 2$ and any given pair $(t,n)$, where $1 \leq t \leq n$, there exists some $\lambda \geq 1$ for which there exists
an OA$_\lambda (t,n,q)$.
\end{lemma}
\begin{proof}
Consider a matrix $\cM$ whose rows are all the $q^n$ distinct words of length $n$ over~$\Z_q$.
Clearly, in each projection of $t$ coordinates from $\cM$ each $t$-tuple is contained in ${\frac{q^n}{q^t}= q^{n-t}}$ distinct rows (codewords).
Thus, the $q^n \times n$ matrix $\cM$ forms an OA$_\lambda (t,n,q)$, where ${\lambda = q^{n-t}}$.
\end{proof}

\begin{lemma}
\label{lem:gen_anticode_preJQ}
Let $\cC_\cD$ be a constant-weight code of length $n$ and weight~$w$ over $\Z_q$, $q>2$,
with distances between the codewords of $\cC_\cD$ taken from a subset $\cD$.
Let $\cA$ be a subset of J$_q(n,w)$ and let $\cC'_\cD \subseteq \cA$ be the largest code in $\cA$ with
distances taken from $\cD$. Then
\begin{equation}
\label{eq:pre_anticode_JQ}
\frac{|\cC_\cD|}{\binom{n}{w} (q-1)^w} \leq \frac{|\cC'_\cD|}{\abs{\cA}} ~.
\end{equation}
\end{lemma}
\begin{proof}
Consider the set of pairs
$$
\cP= \{(c,\pi) ~:~ c \in \cC_\cD, ~ \pi \in S_n, ~ \supp{\pi (c)} = \supp{a},~\text{where}~a \in \cA \}.
$$
For a fixed $c \in \cC_\cD$ and a fixed $a \in \cA$ there are exactly $w! (n-w)!$ choices for $\pi$, for which $\supp{\pi (c)} = \supp{a}$.
Hence, the number of pairs in $\cP$ equals to $\abs{\cC_\cD} \cdot \abs{\cA} \cdot w! \cdot (n-w)!$.

For the word $v=(v_1,v_2,\ldots,v_n) \in \Z_{q-1}^n$, we form a subset $\cA_v$ of J$_q(n,w)$ as follows.
Given a word $x=(x_1,x_2,\ldots,x_n)$ of $\cA$, the word $a_v=(a_1,a_2,\ldots,a_n)$ is constructed in $\cA_v$ as follows.
\begin{enumerate}
\item If $x_i=0$, then $a_i=x_i=0$.

\item If $x_i \neq 0$, then $a_i = x_i +v_i$ when $x_i +v_i < q$ and $a_i = x_i +v_i -(q-1)$ when $x_i +v_i \geq q$.
In other words, if $j=v_i$, then $a_i$ takes the $j$-th nonzero value of $\Z_q$ after the value of~$x_i$, where 1 follows $q-1$.
\end{enumerate}
Using this definition, we have that $\supp{a_v}=\supp{x}$.

Clearly, $\cA_v$ is obtained from $\cA$ by permuting the nonzero elements in each one of the $w$ nonzero coordinates,
of the words in $\cA$, by some $w$ cyclic
permutations (a permutation for each coordinate) on the $q-1$ nonzero symbols of $\Z_q$ (which can be different for each coordinate)
and hence $\abs{\cA_v}=\abs{\cA}$. Moreover, $\cA_v$ and $\cA$ are isomorphic subsets of J$_q(n,w)$.

Now, let $\cM$ be any orthogonal array OA$_\lambda (w,n,q-1)$, for some $\lambda \geq 1$,
whose existence is implied by Lemma~\ref{lem:lambda_for_all}. The number of rows of $M$ is $\lambda (q-1)^w$.

Consider now the set of triples
$$
\cT = \{ (c,\pi,v) ~:~ c \in \cC_\cD, ~ \pi \in S_n,~ v \in \cM, ~ \pi (c) \in \cA_v~  \}.
$$
Let $(c,\pi)$ be a pair in $\cP$, i.e.,
$c \in \cC_\cD$, $\pi \in S_n$, and $\supp{\pi (c)} = \supp{a}$ for some $a \in \cA$. Let $X = \supp{a}$ and let $u=(u_1,u_2,\ldots,u_n)$
be a word in $\Z_{q-1}^n$ such that $a=\pi (c)_u$. It is easy to verify that for each word $v \in \Z_{q-1}^n$, for which the projection
of the coordinates in $X$ on~$u$ and the projection of the coordinates in $X$ on~$v$ are equal, we have that $\pi (c)_v = a = \pi (c)_u$.
Since $\cM$ contains $\lambda$ rows for which these projections are equal, it follows that $\abs{\cT} = \lambda \abs{\cP}$.

Note, that for each permutation $\pi \in S_n$ and two elements $x,y \in \text{J}_q(n,w)$, we have that $d(\pi(x),\pi(y)) = d(x,y)$.
This implies that a fixed permutation $\pi$ with a fixed row $v \in \cM$ can transfer the elements of $\cC_\cD$ into at most $\abs{\cC'_\cD}$
elements of~$\cA_v$. Therefore, the number of triples in $\cT$ is at most $\lambda \cdot \abs{\cC'_\cD} \cdot n! \cdot (q-1)^w$ which implies that
$$
\lambda \cdot \abs{\cC_\cD} \cdot \abs{\cA} \cdot w! \cdot (n-w)! = \lambda \cdot \abs{\cP}
= \abs{\cT}  \leq \lambda \cdot \abs{\cC'_\cD} \cdot n! \cdot (q-1)^w~,
$$
and hence the claim of the lemma is proved.
\end{proof}
The important consequence from Lemma~\ref{lem:gen_anticode_preJQ} is that~(\ref{eq:local_inequality}) and~(\ref{eq:code_anticode_bound})
are satisfied and hence we can consider diameter perfect codes in J$_q(n,w)$.

\subsection{Diameter perfect constant-weight codes for which $w=n$}
\label{sec:w=n}

The first family, [{\bf F1}], of non-binary diameter perfect constant-weight codes contains all perfect code in the Hamming scheme
and all diameter perfect codes in the Hamming scheme. Since by Theorem~\ref{thm:perfect=diameter} any perfect code in the Hamming scheme
is also a diameter perfect code, it follows that we can consider only diameter perfect codes.
In this family of codes we have that $w=n$ and the size of the non-binary alphabet is increased by one compared to the
alphabet in the Hamming scheme. The relation between these families of codes is stated in the following theorem.

\begin{theorem}
\label{thm:increase_alphabet_diam}
There exists a $D$-diameter perfect code of length $n$ over an alphabet of size~${q-1}$ in the Hamming scheme,
if and only if there exists a $D$-diameter perfect constant-weight code of length $n$ and weight $w=n$ over an alphabet of size $q$.
\end{theorem}
\begin{proof}
Let $\cC$ be a $D$-diameter perfect code of length $n$ over the alphabet $\{1,2,\ldots,q-1\}$ in the Hamming scheme.
Let $\cA$ be the related maximum size anticode with diameter $D$ for which $\abs{\cC} \cdot \abs{\cA} = (q-1)^n$.
We define the same code $\cC' \triangleq \cC$ and the same anticode $\cA' \triangleq \cA$,
over the extended alphabet $Q=\{0,1,2,\ldots,q-1\}$. We claim that $\cC'$ is a $D$-diameter perfect constant-weight code of length $n$,
weight $w=n$, and minimum Hamming distance $D+1$, over $Q$. We also claim that $\cA'$ is a maximum size anticode of length $n$, weight $w=n$,
and maximum Hamming diatance $D$, over $Q$. Clearly, the minimum Hamming
distance of $\cC'$ is equal to the minimum Hamming distance of $\cC$, i.e.,~$D+1$. Similarly, the maximum Hamming
distance of $\cA'$ is equal to the maximum Hamming distance of $\cA$, i.e.,~$D$. Moreover,
$$
\abs{\cC'} \cdot \abs{\cA'} = \abs{\cC} \cdot \abs{\cA} = (q-1)^n = \abs{\text{J}_q(n,n)},
$$
which completes the proof of our claims.

Let $\cC$ be an $D$-diameter perfect constant-weight code of length $n$ and weight $w=n$ over the alphabet $Q=\{0,1,2,\ldots,q-1\}$.
Using similar arguments, in reverse order, the same code defined over $Q^*$ is
a $D$-diameter perfect code, of length $n$ over $Q^*$, in the Hamming scheme. Similarly, if $\cA$ is a maximum size anticode,
of length $n$ and weight $w=n$, with diameter~$D$, over $Q$, then the same anticode defined over $Q^*$ is
a maximum size anticode over~$Q^*$.
\end{proof}

In other words, Theorem~\ref{thm:increase_alphabet_diam} implies that when $w=n$, in each word all the coordinates are nonzero.
Hence, the words in J$_q(n,n)$ are over an alphabet with only $q-1$ letters. It implies that any code in J$_q(n,n)$ can be considered
as a code in the Hamming scheme over an alphabet with $q-1$ letters. Therefore, any $D$-diameter perfect code
of length $n$ over an alphabet with $q-1$ letters is also $D$-diameter perfect code in J$_q(n,n)$.
Similarly, each maximum size anticode of length $n$ and diameter $D$ over an alphabet with $q-1$ letters (with no \emph{zeroes})
is also a maximum size anticode with diameter $D$ in J$_q(n,n)$. Finally, all the diameter perfect codes in the Hamming scheme
over an alphabet whose size is a prime power were characterized in~\cite{AAK01}.

\begin{theorem}
\label{thm:diam_Hamming}
In the Hamming scheme whose alphabet size is a prime power there are no diameter perfect codes except for
the Hamming codes, the extended Hamming codes, the Golay codes, the extended Golay codes, and the MDS codes.
\end{theorem}

When $q$ is not a prime power the only known diameter perfect codes are the orthogonal arrays with index unity.

\subsection{Codes with alphabet size $2^k+1$ for which $w=n-1$}
\label{sec:perfect}

By Theorem~\ref{thm:eP_diamP} an $e$-perfect code in J$_q(n,w)$ is also a $(2e)$-diameter perfect code in J$_q(n,w)$.
All known nontrivial non-binary perfect constant-weight codes of length $n$ have weight $w=n-1$ and they form an important class
of the family [{\bf F2}] of non-binary diameter perfect codes.
There are a few known families of nontrivial non-binary perfect constant-weight codes.
Ternary 1-perfect codes in J$_3(2^m,2^m-1)$, $m \geq 2$, were constructed in~\cite{Sva99,LiTo99} and it was proved in~\cite{LiTo99}
that there are no other ternary 1-perfect constant-weight codes. A large number of nonequivalent codes with
these parameters were constructed by~\cite{Kro08}. The construction in~\cite{LiTo99} was
generalized in~\cite{EtLi01} and 1-perfect codes were constructed in J$_q(q+1,q)$, where $q=2^k +1$, $k \geq 2$.
It is not known whether there exist more nontrivial perfect constant-weight codes.

Ternary diameter perfect codes of length $n$ and weight $w=n-1$ were considered in~\cite{Kro08}.
For $d=4$, it was proved in~\cite{Kro08} that there is only one set of parameters for which there exist
a ternary 3-diameter perfect constant-weight codes of length $n$ and weight $w=n-1$. A ternary code in this set has length 6, weight 5,
and 12 codewords. Such a code was constructed in~\cite{OsSv02}.
For ternary 4-diameter perfect constant-weight codes, it was proved by~\cite{Kro08} that for every length $n=2^m$, $m$ odd, there
exists such a code with weight $w=n-1$. When $m$ is even only a code
of length $n=64$ constructed in~\cite{BDMW} is known~\cite{KOP16} and it was proved in~\cite{KOP16} that
such a code does not exist for length $n=16$.

\subsection{Generalized Steiner systems}
\label{sec:GST}

We already saw that a Steiner system S$(t,w,n)$ is a binary $(w-t)$-diameter perfect constant-weight code.
For non-binary alphabet, there is an analog definition of generalized Steiner system which was introduced in~\cite{Etz97}.

\begin{definition}
A {\bf \emph{generalized Steiner system}} GS$(t,w,n,q)$ is a constant-weight code~$\cC$, over $\Z_q$,
whose length is $n$, weight $w$, for each codeword, such that:
\begin{enumerate}
\item The minimum Hamming distance of $\cC$ is $2(w-t)+1$.

\item Each word $x$ of length $n$ and weight $t$ over $\Z_q$ is covered by exactly one codeword $c \in \cC$,
i.e., $d(x,c)=w-t$.
\end{enumerate}
\end{definition}

Similarly to Theorem~\ref{thm:ST_b_s} we have the following theorem.
\begin{theorem}
\label{thm:num_blocksG}
The number of codewords in a generalized Steiner system GS$(t,w,n,q)$ is
$$
\frac{\binom{n}{t}}{\binom{w}{t}} (q-1)^t
$$
and its minimum Hamming distance is $2(w-t)+1$.
\end{theorem}

Similarly to Lemma~\ref{lem:derived} one can verify the following lemma.
\begin{lemma}
\label{lem:derivedGS}
If there exists a GS$(t,w,n,q)$, $t > 1$, then there exists a GS$(t-1,w-1,n-1,q)$.
\end{lemma}

Let $\cA^\text{s} (n,w,t)$ be the anticode defined by
$$
\cA^\text{s} (n,w,t) \triangleq \{ (\overbrace{1\cdots \cdots 1}^{t ~ \text{times}}, a_1,\ldots,a_{n-t} ) ~:~ a_i \in \Z_q, ~ \wt{a_1 \cdots a_{n-t}} = w-t \}~.
$$
Note, that when $q=2$ we have that $\cA^\text{s} (n,w,t)$ is identical to $\cA (n,w,t)$.

The following lemma can be readily verified.
\begin{lemma}
\label{lem:anticode_sizeGS}
The anticode $\cA^\text{s} (n,w,t)$ has codewords of length $n$ and weight~$w$, with maximum distance $2(w-t)$, where $n \leq 2w-t$, over $\Z_q$.
The anticode $\cA^\text{s} (n,w,t)$ has $\binom{n-t}{w-t} (q-1)^{w-t}$ codewords.
\end{lemma}

\begin{lemma}
\label{lem:max_anticodeGS}
If there exists a generalized Steiner system S$(t,w,n,q)$, then the anticode $\cA^\text{s} (n,w,t)$ is a maximum size anticode
of length~$n$ and weight $w$, with maximum distance $2(w-t)$, where $n \leq 2w-t$, over $\Z_q$.
\end{lemma}
\begin{proof}
Let $\cC$ be a generalized Steiner system GS$(t,w,n,q)$ and let $\cA$ be the anticode $\cA^\text{s} (n,w,t)$.
By the definition of a generalized Steiner system and by Lemma~\ref{thm:num_blocksG}, $\cC$ has minimum
Hamming distance $2(w-t)+1$ and its size is $\frac{\binom{n}{t}}{\binom{w}{t}} (q-1)^t$. By Lemma~\ref{lem:anticode_sizeGS} the
anticode $\cA^\text{s} (n,w,t)$ has maximum distance $2(w-t)$ and its size is $\binom{n-t}{w-t} (q-1)^{w-t}$. Since
$$
\abs{\cC} \cdot \abs{\cA} = \frac{\binom{n}{t}}{\binom{w}{t}} (q-1)^t \cdot
\binom{n-t}{w-t} (q-1)^{w-t} = \binom{n}{w} (q-1)^w =\abs{\text{J}_q(n,w)} ~,
$$
it follows by the code-anticode bound that $\cA^\text{s} (n,w,t)$ is a maximum size anticode of length~$n$ and
weight $w$, over $\Z_q$, whose maximum distance~${2(w-t)}$.
\end{proof}
\begin{corollary}
A generalized Steiner system GS$(t,w,n,q)$ is a ${2(w-t)}$-diameter perfect code.
\end{corollary}

Generalized Stiener systems were defined in~\cite{Etz97} and further considered in many papers, e.g.~\cite{CJZ07,CGZ99,PhYi97,WuFa09},
but there is still lot ground for further research in this area.

\subsection{Maximum distance separable constant weight codes}
\label{sec:MDS-CW}

\begin{definition}
An $(n,w,q)$ {\bf \emph{maximum distance separable constant-weight code}} (MDS-CW code in short) is a
constant-weight code~$\cC$, over $\Z_q$, whose length is~$n$, weight $w$, for each codeword, such that:
\begin{enumerate}
\item The minimum Hamming distance of $\cC$ is $w$.

\item Each subset of $w$ coordinates is the support of exactly $q-1$ codewords.
\end{enumerate}
\end{definition}
The name MDS-CW codes is a consequence from the observation that the minimum weight codewords
in an MDS codes form such an MDS-CW code~\cite[pp. 319--321]{McSl77}.

Let $\cA^\text{m} (n,w,\delta)$, $1 \leq \delta \leq w$, be the anticode defined as follows
$$
\cA^\text{m} (n,w,\delta) \triangleq \{ (a_1,a_2, \ldots, a_\delta , \overbrace{1\cdots \cdots 1}^{w-\delta ~ \text{times}} ,
\overbrace{0\cdots \cdots 0}^{n-w ~ \text{times}}) ~:~ a_i \in \Z_q \setminus \{0\} \}~.
$$

The following lemma can be readily verified.
\begin{lemma}
\label{lem:anticode_sizeMDSCW}
The anticode $\cA^\text{m} (n,w,\delta)$ has codewords of length $n$, weight $w$, with maximum distance $\delta$,
and the number of codewords in $\cA^\text{m} (n,w,\delta)$ is $(q-1)^\delta$.
\end{lemma}

\begin{lemma}
\label{lem:max_anticodeMDSCW}
If there exists an $(n,w,q)$ MDS-CW code, then the anticode $\cA^\text{m} (n,w,w-1)$ is a maximum size anticode of length $n$,
weight $w$, with maximum distance $w-1$, over $\Z_q$.
\end{lemma}
\begin{proof}
Let $\cC$ be an $(n,w,q)$ MDS-CW code and $\cA$ be the anticode $\cA^\text{m} (n,w,w-1)$.
By the definition, an $(n,w,q)$ MDS-CW code, has minimum distance $w$ and size $\binom{n}{w} (q-1)$.
By Lemma~\ref{lem:anticode_sizeMDSCW}, the anticode $\cA^\text{m} (n,w,w-1)$ has maximum distance $w-1$ and size $(q-1)^{w-1}$. Since
$$
\abs{\cC} \cdot \abs{\cA} =\binom{n}{w} (q-1) \cdot (q-1)^{w-1} = \binom{n}{w} (q-1)^w =\abs{\text{J}_q(n,w)} ~,
$$
it follows by the code-anticode bound that the anticode $\cA^\text{m} (n,w,w-1)$ is a maximum size anticode of length~$n$, weight $w$,
with maximum distance $w-1$, over $\Z_q$.
\end{proof}
\begin{corollary}
An $(n,w,q)$ MDS-CW code is a $(w-1)$-diameter perfect code.
\end{corollary}

Given a pair $(w,n)$ it is very simple to verify from Theorem~\ref{thm:MDS_conjecture} that there exists a prime power $q$ for which
there exists an $(n,w,q)$ MDS-CW code. But, this can be further improved as It was proved in~\cite{Etz97} that there exists a $q_0 =QMDS(n,w)$
such that for each $q \geq q_0$ there exists an $(n,w,q)$ MDS-CW code.
Finally, we will present a theorem related to a very simple union construction~\cite{Etz97} which can be used to prove
an upper bound on this $QMDS(n,w)$.
\begin{theorem}
\label{thm:const_union}
If there exists an $(n,w,q_1)$ MDS-CW code and an $(n,w,q_2)$ MDS-CW code, then there exists an $(n,w,q_1+q_2 -1)$ MDS-CW code.
\end{theorem}

There are more constructions and bounds on the parameters of MDS-CW codes which were presented in~\cite{Etz97}.
Some of these bounds and constructions can be modified and generalized to the sixth family, [{\bf F6}], which will
be discussed in Section~\ref{sec:MOA-CW}.

\subsection{Codes for which $d=w+1$}
\label{sec:d=w+1}

When $d=w+1$ we are looking for an $(n,w+1,w)_q$ code, i.e., a constant-weight code of length $n$, weight $w$, and minimum
Hamming distance $w+1$, over $\Z_q$. In such a code, each subset of $w$ coordinates will support exactly one codeword,
i.e., the number of codewords is~$\binom{n}{w}$. It is rather easy to verify that such a code is a
$w$-diameter perfect constant-weight code and it exists for any given $n$ and $w$ as it is proved in the following theorem.

\begin{theorem}
\label{thm:alphabet_d=w+1}
If $n$ and $w$ are integers such that $1 \leq w \leq n-1$, then there exists a $q_0 (w,n)$ such that for each
$q \geq q_0 (w,n)$ there exist an $(n,w+1,w)_q$ $w$-diameter perfect code $\cC$.
\end{theorem}
\begin{proof}
First, note that since the minimum distance of an $(n,w+1,w)_q$ code $\cC$ is $w+1$, it follows that each subset of $w$ coordinates can
be a support for at most one codeword. If each such subset of $w$ coordinates supports exactly one codeword, then the total number of codewords
in $\cC$ will be $\binom{n}{w}$. Assume further that in $\cC$ for each coordinate all the nonzero elements in the codewords of $\cC$
have distinct symbols. It implies that in each coordinate there are $\binom{n}{w} \frac{w}{n} =\binom{n-1}{w-1}$ nonzero symbols.
Let $q' \triangleq 1 + \binom{n-1}{w-1}$ and let $Q$ be an alphabet with $q=q' +\epsilon$ symbols, where $\epsilon \geq 0$.
Assign now for each coordinate a different nonzero symbols from the $q' + \epsilon -1$ nonzero symbols of $Q^*$ to each codeword
that has a nonzero symbol in this coordinate. Clearly, $\cC$ in an $(n,w+1,w)_q$ code with $\binom{n}{w}$ codewords.

Let $\cA$ be the anticode $\cA^\text{m}(n,w,w)$ over $Q$.
By Lemma~\ref{lem:anticode_sizeMDSCW}, the anticode $\cA$, has diameter~$w$ and $(q-1)^w$ codewords. Clearly,
$$
\abs{\cC} \cdot \abs{\cA} = \binom{n}{w} (q-1)^w = \abs{J_q (n,w)}
$$
and hence by the code-anticode bound, $\cC$ is an $(n,w+1,w)_q$~~$w$-diameter perfect constant-weight code over the alphabet $Q$ of size $q$.
\end{proof}
\begin{corollary}
If there exists an $(n,w+1,w)_q$ code with $\binom{n}{w}$ codewords, then $\cA^\text{m}(n,w,w)$ is a maximum size anticode
of length $n$, weight $w$, and diameter $w$, over and alphabet with $q$ symbols.
\end{corollary}

The proof of Theorem~\ref{thm:alphabet_d=w+1} implies that an $(n,w+1,w)_q$ $w$-diameter constant-weight perfect code $\cC$
has $\binom{n}{w}$ codewords, where each $w$-subset of $w$ coordinates of $\cC$ is the support of exactly one codeword of $\cC$.
In view of Theorem~\ref{thm:alphabet_d=w+1} our goal now is to find $q_0(w,n)$ which is the smallest size alphabet $q$ for such an
$(n,w+1,w)_q$ code exists.

\begin{corollary}
\label{cor:F4trivial}
For each alphabet $Q$ of size $q$, where $q \geq 1 + \binom{n-1}{w-1}$, there exists an $(n,w+1,w)_q$ code,
i.e., $q_0 (w,n) \leq 1 + \binom{n-1}{w-1}$.
\end{corollary}

\begin{lemma}
For each $w \geq 1$ there exists an $(w+1,w+1,w)_{w+1}$ code which is a $w$-diameter perfect code.
\end{lemma}
\begin{proof}
Follows immediately from the fact that if there is a codeword on each subset of $w$ coordinates, then there are exactly $w$
codewords with nonzero symbols on each coordinates.
\end{proof}
\begin{corollary}
If $w \geq 1$, then $q_0 (w,w+1) = w+1$.
\end{corollary}

Theorem~\ref{thm:alphabet_d=w+1} implies the existence of an $(n,w+1,w)_q$ code for each $q \geq q_0 (w,n)$, but the
upper bound $1 + \binom{n-1}{w-1}$ on $q_0 (w,n)$, inferred in Corollary~\ref{cor:F4trivial}, is quite large. Can we find a better upper
bound on $q_0 (w,n)$? The answer is definitely positive and for this purpose we have the following results.

\begin{lemma}
If $n > w+1$, then $q_0 (w,n) \geq q_0 (w,n-1)$.
\end{lemma}
\begin{proof}
Assume that $\cC$ is an $(n,w+1,w)_q$ $w$-diameter constant-weight perfect code and let~$S$ be any subset of $n-1$ coordinates.
By definition, the set codewords whose supports are subsets of $S$ form an $(n-1,w+1,w)_q$ $w$-diameter constant-weight perfect code.
Thus, the claim of the lemma follows.
\end{proof}
\begin{corollary}
If $n > w+1$ then $q_0 (w,n) \geq w+1$.
\end{corollary}

\begin{lemma}
\label{lem:tr_bound}
If $n > w+1$, then $q_0 (w,n) \geq n-w+2$.
\end{lemma}
\begin{proof}
Let $\cC$ be an $(n,w+1,w)_q$ $w$-diameter perfect code $\cC$.
Consider the sub-code $\cC'$ of codewords from $\cC$ for which there is no \emph{zero} in the first $w-1$ coordinates.
Since, each one of the other $n-w+1$ coordinates must have a nonzero symbol with exactly one of these codewords, it follows that
the sub-code $\cC'$ contains $n-w+1$ codewords. Since the distance of $\cC'$ is $w+1$, each pair of codewords of $\cC'$ have only
two distinct coordinates in their supports, and each pair of codewords of $\cC'$ have $w-1$ joint coordinates with nonzero symbols,
it follows that in each given coordinate of the first $w-1$ coordinates the codewords of $\cC'$ have
distinct nonzero symbols. Since $\abs{\cC'} = n-w+1$, it follows that $\cC$ has at least $n-w+1$ nonzero symbols and hence $q \geq n-w+2$.
\end{proof}
\begin{corollary}
\label{cor:w=2F4}
If $n > 2$, then $q_0 (2,n)=n$.
\end{corollary}
\begin{proof}
By Lemma~\ref{lem:tr_bound} we have that $q_0 (2,n) \geq n$ and by Corollary~\ref{cor:F4trivial} we have that ${q_0 (2,n) \leq n}$.

Thus, $q_0 (2,n)=n$.
\end{proof}

The proof of the next theorem requires two more concepts, a one-factorization and a near-one-factorization.
A \emph{one-factorization} of the complete graph $K_n$, $n$ even, is a partition of the
edges of $K_n$ into perfect matchings. In other words, the set
$$
\cF = \{ \cF_1, \cF_2 , \ldots , \cF_{n-1} \}
$$
is a one-factorization of $K_n$ if each $\cF_i$, $1 \leq i \leq n-1$, is a perfect matching
(called a \emph{one-factor}), and the $\cF_i$'s are pairwise disjoint.

If $n$ is odd, then there is no perfect matching in $K_n$ and we define a \emph{near-one-factorization}
$$
\cF = \{ \cF_1, \cF_2 , \ldots , \cF_n \}
$$
to be a partition of the edges in $K_n$ into sets of $\frac{n-1}{2}$ pairwise disjoint edges,
where each $\cF_i$ has one isolated vertex. Each $\cF_i$ is called a \emph{near-one-factor}.

\begin{theorem}
\label{thm:q0_3n}
If $n$ is odd, then $q_0 (3,n)=n-1$, and if $n$ is even, then $q_0 (3,n)=n$.
\end{theorem}
\begin{proof}
By Lemma~\ref{lem:tr_bound} we have that $q_0 (3,n) \geq n-1$ and this bound is applied when $n$ is odd.

Assume now that $n$ is even and let $\cC$ be a related code. Let $\cC_1$ be the set of codewords in $\cC$ with a nonzero symbol in the
first coordinate. By the definition of this family of codes, it follows that $\abs{\cC_1}=\binom{n-1}{2} = \frac{(n-1)(n-2)}{2}$.
Since the minimum distance of $\cC_1$ is four and $n$ is even, it follows that the number of codewords in $\cC_1$ with a given nonzero
symbol $\sigma$ in the first coordinate is at most $\frac{n-2}{2}$. Since $\abs{\cC_1} = \frac{(n-1)(n-2)}{2}$,
it follows that there are at least $n-1$ nonzero symbols in the first coordinate.
Therefore, $q_0 (3,n) \geq n$ if $n$ is even.

Regarding the upper bound on $q_0(3,n)$ we distinguish again between two cases, depending on whether $n$ is odd or $n$ is even.

\noindent{\bf Case 1.}
$n$ is odd.

Let $\cN$ be the set of $n$ coordinates and let $Q \triangleq \{ 0, \sigma_1,\ldots,\sigma_{n-2} \}$ be an alphabet of size~$n-1$.
Let $\cC$ be a code of length $n$ and weight 3 with $\binom{n}{3}$ codewords, a codeword for each support of size 3.
Consider the $i$-th coordinate, $i \in \cN$ and
let $\cF = \{ \cF_1 , \cF_2 , \ldots, \cF_{n-2} \}$ be a one-factorization on the $n-1$ points of $\cN \setminus \{i\}$.
Given a triple $\{ i, j,k \}$, where $\{ j,k \} \in \cF_r$, we assign $\sigma_r$ to the symbol in coordinate $i$ of the
codeword $\{ i, j,k \}$, where the nonzero symbols are in coordinates $i,j,k$. It is readily verified that we have constructed
a code of length $n$ and weight 3, over an alphabet $Q$ of size $n-1$. Clearly, if two codewords share at most one coordinate,
then their Hamming distance is at least 4. Now, assume that two codewords $c_1$ and $c_2$ share nonzero symbols in two coordinates $i$ and $j$.
If the symbols in the $i$-th coordinate of $c_1$ and $c_2$ are distinct and the symbols in the $j$-th coordinate of $c_1$ and $c_2$ are distinct,
then clearly $d(c_1,c_2)=4$. Now, assume for the contrary that in one coordinate, say~$i$, $c_1$ and~$c_2$ have the same symbol.
By the construction, we have that the two other pairs of nonzero coordinates in $c_1$ and~$c_2$ must be disjoint (they belong to the
same one-factor), a contradiction. Therefore, the minimum distance of $\cC$ is $4$ and hence $q_0 (3,n) \leq n-1$.

\noindent{\bf Case 2.}
$n$ is even.

Let $\cN$ be the set of $n$ coordinates and let $Q \triangleq \{ 0, \sigma_1,\ldots,\sigma_{n-1} \}$ be an alphabet of size~$n$.
Let $\cC$ be a code of length $n$ and weight 3 with $\binom{n}{3}$ codewords, a codeword for each support of size 3.
Consider the $i$-th coordinate, $i \in \cN$ and
let $\cF = \{ \cF_1 , \cF_2 , \ldots, \cF_{n-1} \}$ be a near-one-factorization on the $n-1$ points of $\cN \setminus \{i\}$.
Given a triple $\{ i, j,k \}$, where $\{ j,k \} \in \cF_r$, we assign $\sigma_r$ to the symbol in coordinate $i$ of the
codeword $\{ i, j,k \}$. It is readily verified that we have constructed
a code of length $n$ and weight 3, over an alphabet $Q$ of size $n$. As in Case 1 the minimum distance of $\cC$ is 4 and
therefore, $q_0 (3,n) \leq n$.

Thus, these two cases complete the proof of the theorem.
\end{proof}

Other codes with the same parameters as the ones constructed in Theorem~\ref{thm:q0_3n} were also presented in~\cite{CDLL08,ChLi07},
by using different techniques from combinatorial designs.
Similarly to the technique used in the proof of Theorem~\ref{thm:q0_3n} one can construct ${(n,w+1,w)_q}$ $w$-diameter perfect codes,
for relatively small $q$, when $w$ is small using techniques coming from combinatorial designs. The same is true for
$(n,w+1,w)_q$ $w$-diameter perfect codes, when $n$ is not much larger than $w$. Such constructions are left for future research.
Moreover, the technique used in Theorem~\ref{thm:q0_3n} to obtain the upper bound on $q_0 (3,n)$ can be used to obtain
better upper bounds on $q_0 (w+1,n+1)$ than the trivial one, i.e., $q_0 (w+1,n+1) \leq 1 + \binom{n}{w}$. The idea is to
partition the set of all binary words of length $n$ and weight $w$ into pairwise disjoint constant-weight codes of length $n$,
weight $w$, and minimum Hamming distance $w+2$. Let~$\chi (n,w)$ be the minimum number of codes in such a partition.
With an identical proof as the one in Theorem~\ref{thm:q0_3n} we can prove that $q_0 (w+1,n+1) \leq \chi (n,w) +1$.
Note, that the minimum distance of a constant-weight code is always even and hence the proof will be more effective for even $w$,
i.e., for bounds on $q_0(w+1,n+1)$ when $w+1$ is odd.
This kind of a partition problem was discussed in~\cite{BSSS}. As a simple example we can use prove the following bound implied
by a result proved in~\cite{GrSl80}.
\begin{theorem}
If $w$ is even and $p$ is the smallest prime power for which $p \geq n$, then
$$
q_0 (w+1,n+1) \leq 1 + p^{w/2}~.
$$
\end{theorem}
We leave further improvements in this direction for future research.

\subsection{Multiple orthogonal arrays constant weight codes}
\label{sec:MOA-CW}

In the last family of diameter perfect constant-weight codes we have similarly as in families [{\bf F4}] and [{\bf F5}] that each projection
of any $w$ coordinates supports some specified number of codeword codewords. The distinction from families [{\bf F4}] and [{\bf F5}] is that
the minimum distance of the code in this family, [{\bf F6}] , is strictly smaller than the weight of the codewords.
More precisely we have the following definition.

\begin{definition}
An $(n,d,w)_q$ {\bf \emph{multiple orthogonal arrays constant weight}} (MOA-CW in short) code is a code of
length $n$, constant weight $w$, minimum distance $d <w$, where each subset of $w$ coordinates is the support of
exactly $(q-1)^{w-d+1}$ codewords, i.e., these codewords form an OA$(w-d+1,w,q-1)$.
\end{definition}

One might asks why this family does not include the MDS-CW codes, where $d=w$. The two families
of codes, [{\bf F4}] and [{\bf F6}] share some properties such as similar expression of their size (which is also shared
by the family [{\bf F5}]), they are both related to orthogonal arrays (MDS codes) in a way that codewords with no \emph{zeroes}
in the same $w$ coordinates form an orthogonal array. The main reason for the distinction
between the two families is that an MDS-CW code either forms the codewords of minimum weight in an orthogonal array
or has the same parameters as it would have had, if such an orthogonal array have been possible. There is no similar property
for an MOA-CW code. The codewords of an MOA-CW code
are not related to codewords of some weights in MDS codes or orthogonal arrays. Another important distinction
is in the simple union construction of Theorem~\ref{thm:const_union} which cannot be applied to MOA-CW codes.
The similarity of the two families will be also demonstrated
in one construction of such codes which is a joint construction for both families of codes.
Similarly, some bounds on the tradeoff between the parameters of these codes are joint bounds for the two families of codes.

\begin{theorem}
An $(n,d,w)_q$ MOA-CW code is a $(d-1)$-diameter perfect constant-weight code.
\end{theorem}
\begin{proof}
By definition, the size of an $(n,d,w)_q$ MOA-CW code $\cC$ is $\binom{n}{w} (q-1)^{w-d+1}$ and by Lemma~\ref{lem:anticode_sizeMDSCW}
the related anticode~$\cA$ with
diameter $d-1$ in J$_q(n,w)$, $\cA^m (n,w,d-1)$, has size $(q-1)^{d-1}$. Therefore we have that,
$$
\abs{\cC} \cdot \abs{\cA} = \binom{n}{w} (q-1)^{w-d+1} (q-1)^{d-1} = \binom{n}{w} (q-1)^w = \abs{J_q(n,w)}~,
$$
which by the code-anticode bound implies that $\cC$ is a $(d-1)$-diameter perfect constant-weight code.
\end{proof}
\begin{corollary}
If there exists an $(n,d,w)_q$ MOA-CW code, then $\cA^\text{m}(n,w,d-1)$ is a maximum size anticode
of length $n$, weight $w$, and diameter $d-1$, over an alphabet with $q$ symbols.
\end{corollary}

Next, we present a construction for MOA-CW codes which can also serve as a construction for MDS-CW codes.
The construction is a generalization and a modification of a construction presented in~\cite{Etz97}.
Let $\cM$ be an OA$(t,n,q)$ over $Q$, where $Q=\{1,2,\ldots,q \}$. Assume further that the symbols in the last coordinate
of the first $q^{t-1}$ codewords of $\cM$ are \emph{ones}, the symbols in the last coordinate of the next $q^{t-1}$ codewords
of $\cM$ are \emph{twos}, and so on, where the symbols in the last coordinate of the last $q^{t-1}$ codewords of $\cM$ are $q$'s.
Assume further that ${q \geq \binom{n-1}{\ell}}$ for a given $\ell$, $1 \leq \ell \leq n-1$.
Let $S_1,S_2,\ldots,S_r$, where $r = \binom{n-1}{\ell}$, be a sequence containing all the $\ell$-subsets
of $\{1,2,\ldots,n-1\}$. Let $\cM'$ be the $\left( \binom{n-1}{\ell} q^{t-1} \right) \times (n-1)$ array constructed from $\cM$ as follows.
\begin{enumerate}
\item If $S_1 = \{ i_1, i_2, \ldots , i_\ell \}$, then replace all the symbols in the first $q^{t-1}$ rows of
column $i_j$ in~$\cM$, for each $1 \leq j \leq \ell$, with \emph{zeroes}.

\item If $S_2 = \{ i_1, i_2, \ldots , i_\ell \}$, then replace all the symbols in the next $q^{t-1}$ rows of
column $i_j$ in~$\cM$, for each $1 \leq j \leq \ell$, with \emph{zeroes}.

\item Continue the same process with $S_3$, $S_4$, and so on until $S_r$.

\item Remove the last column of $\cM$.

\item Remove the last $q^t - \left( \binom{n-1}{\ell} q^{t-1} \right)$ rows of $\cM$.
\end{enumerate}

\begin{theorem}
\label{thm:main_MOA-CW}
The rows of the array $\cM'$ form an $(n-1,n-t-\ell+1,n-\ell-1)_{q+1}$ code~$\cC$ that is an $(n-t-\ell)$-diameter
perfect constant-weight code over $Q \cup \{ 0 \}$.
\end{theorem}
\begin{proof}
Exactly one column was deleted from $\cM$ to obtain $\cM'$ and hence the length of the code $\cM'$ is $n-1$. In each codeword of length $n-1$
exactly $\ell$ \emph{zeroes} were inserted instead of nonzero symbols and hence the weight of each codeword is $n-1-\ell$.
Since $\cM$ is an OA$(t,n,q)$, it follows that
$\abs{\cM} =q^t$, and since $q \geq \binom{n-1}{\ell}$, it follows that $q^t \geq \binom{n-1}{\ell} q^{t-1}$ and hence $\cM$ has at
least $\binom{n-1}{\ell} q^{t-1}$ rows as required by the construction. Furthermore, note that the minimum distance of
the code defined by $\cM$ is $n-t+1$.
Let $c_1$ and $c_2$ be two codewords in $\cM'$. If the \emph{zeroes} in $c_1$ and $c_2$ are on the same $\ell$ coordinates, then
$c_1$ and $c_2$ were derived from two rows $c'_1 \alpha$ and $c'_2 \alpha$ of $\cM$, where $\alpha \in \{1,2,\ldots,q\}$,
and $d(c'_1 \alpha ,c'_2 \alpha) \geq n-t+1$. Since the same $\ell$ coordinates were changed in $c'_1$ and $c'_2$ to obtain $c_1$ and $c_2$,
respectively, it follows that $d(c_1,c_2) \geq n-t+1-\ell$. If the \emph{zeroes} in $c_1$ and $c_2$ are not on the same coordinates, then
$c_1$ and $c_2$ were derived from two rows $c'_1 \alpha$ and $c'_2 \beta$, where $\alpha, \beta \in \{1,2,\ldots,q\}$, $\alpha \neq \beta$,
and $d(c'_1 \alpha , c'_2 \beta) \geq n-t+1$, which implies that $d(c'_1,c'_2) \geq n-t$. The number of coordinates in which both $c_1$ and $c_2$
have \emph{zeroes} is at most $\ell-1$ and hence $d(c_1,c_2) \geq d(c'_1,c'_2) -(\ell-1) \geq n-t -(\ell-1)= n-t-\ell+1$.
Thus, $d(\cC) \geq n-t-\ell+1$.

As an immediate consequence from the construction, the number of rows in the array~$\cM'$ is $\binom{n-1}{\ell} q^{t-1}$
and its alphabet $\{ 0,1,2,\ldots,q\}$ is of size $q+1$. Let $\cA$
be a related anticode of length $n-1$ and diameter $n-t-\ell$. By Lemma~\ref{lem:anticode_sizeMDSCW} there exists such a
constant-weight anticode $\cA^\text{m} (n-1,n-1-\ell,n-t-\ell)$, over~$\Z_{q+1}$,
whose size is $q^{n-t-\ell}$. Therefore,
$$
\abs{\cC} \cdot \abs{\cA} = \binom{n-1}{\ell} q^{t-1} \cdot q^{n-t-\ell}
= \binom{n-1}{n-1-\ell} q^{n-\ell-1} = \abs{\text{J}_{q+1}(n-1,n-\ell-1)} ~,
$$
which implies by the code-anticode bound that $\cM'$ is an $(n-t-\ell)$-diameter perfect code.
\end{proof}

\begin{corollary}
When $t=2$ the code $\cM'$ is an $(n-1,n-\ell-1,q+1)$ MDS-CW code.
\end{corollary}
\begin{corollary}
When $t>2$ the code $\cM'$ is an $(n-1,n-t-\ell-1,n-\ell-1)_{q+1}$ MOA-CW code.
\end{corollary}

\begin{theorem}
$~$
\begin{enumerate}
\item If there exists an $(n,d,w)_q$ MOA-CW code, then there exists an $(n-1,d,w)_q$ MOA-CW code.

\item If there exists an $(n,d,w)_q$ MOA-CW code, then there exists an $(n-1,d-1,w-1)_q$ MOA-CW code.
\end{enumerate}
\end{theorem}
\begin{proof}
Let $\cC$ be $(n,d,w)_q$ MOA-CW code and define the following two code
$$
\cC_1 \triangleq \{ (x_2,x_3,\ldots,x_n) ~:~ (0,x_2,x_3,\ldots,x_n) \in \cC \}
$$
and
$$
\cC_2 \triangleq \{ (x_2,x_3,\ldots,x_n) ~:~ (x_1,x_2,x_3,\ldots,x_n) \in \cC, ~ x_1 \neq 0 \}~.
$$
One can easily verify that $\cC_1$ is an $(n-1,d,w)_q$ MOA-CW code and $\cC_2$ is an $(n-1,d-1,w-1)_q$ MOA-CW code.
\end{proof}

After constructing $(d-1)$-diameter constant-weight perfect codes for $d < w$, where each $w$ coordinates are
the support of exactly $(q-1)^{w-d+1}$ codewords we would like to have some lower bounds on the alphabet size of such codes and upper bounds
on their length and their weights. Since each $w$ coordinates are the support of exactly $(q-1)^{w-d+1}$ codewords,
it follows that the projection on each $w$ coordinates on these codewords forms an orthogonal array OA$(w-d+1,w,q-1)$ and the related
bounds on orthogonal arrays in Theorems~\ref{thm:b_OA1},~\ref{thm:b_OA2}, and~\ref{thm:b_OA3}, can be applied. This implies the following theorem
which present a tradeoff between the alphabet size and the minimum distance of the code.

\begin{theorem}
\label{thm:bounds_MOA-CW}
$~$
\begin{enumerate}
\item If there exists an $(n,w-1,w)_q$ MOA-CW code, then $w \leq q$.

\item If there exists an $(n,w-\delta,w)_q$ MOA-CW code, where $2 \leq \delta \leq w-1$ and $q$ is even, then $w \leq q+\delta$.

\item If there exists an $(n,w-\delta,w)_q$ MOA-CW code, where $2 \leq \delta \leq w-1$ and $q$ is odd, then $w \leq q+\delta-1$.

\item If there exists an $(n,w-\delta,w)_q$ MOA-CW code, where $q-1 \leq \delta +1$, then $w \leq \delta +2$.
\end{enumerate}
\end{theorem}
\begin{proof}
All the claims are direct consequences of Theorems~\ref{thm:b_OA1},~\ref{thm:b_OA2}, and~\ref{thm:b_OA3}, where the length $n$
in the OA$(t,n,q)$ is restricted to $w$, the alphabet size is $q-1$, and $w-\delta = n-t+1$.
\end{proof}

Theorem~\ref{thm:bounds_MOA-CW} implies upper bounds on $w$ as a function of the alphabet size $q$ and the minimum distance $d$ of the code
and, similarly, lower bounds on $q$ as a function of the weight~$w$ and the minimum distance of the code. Since $d=w-\delta$, it follows that
these bounds can be written as bounds only as a tradeoff between $d$ and $q$.

\begin{corollary}
$~$
\begin{enumerate}
\item If there exists an $(n,d,w)_q$ MOA-CW code, where $1 \leq d \leq w-2$ and $q$ is even, then $d \leq q$.

\item If there exists an $(n,d,w)_q$ MOA-CW code, where $1 \leq d \leq w-2$ and $q$ is odd, then $d+1 \leq q$.
\end{enumerate}
\end{corollary}

The next bound presents a tradeoff between the length, the weight, and the alphabet size, of the code. It is interesting to note that
the minimum distance has no influence on the bound.
\begin{theorem}
\label{thm:tradeoff_mds-cw}
If there exists an $(n,d,w)_q$ MOA-CW code, then $n \leq q+w-2$ (equivalently, $q \geq n-w+2$).
\end{theorem}
\begin{proof}
Let $\cC$ be an $(n,d,w)_q$ MOA-CW code and consider the set $\cS$ of codewords in $\cC$ which have only nonzero symbols
in the first $w-1$ coordinates and in these $w-1$ coordinates of~$\cS$, all the codewords of $\cS$  share the same
suffix of length $w-d+1$.
Each two such codewords of $\cS$ can differ in at most two coordinates out of the last $n-w+1$ coordinates and in the
first $d-2$ coordinates. Hence, since the minimum distance of $\cC$ is $d$, it follows that
two such codewords of $\cS$ differ exactly in these $d$ coordinates. This implies the following observations:

\begin{enumerate}
\item The only nonzero symbol in the last $n-w+1$ coordinates of each codeword from~$\cS$ is in a distinct coordinate.
This implies that $\abs{S} \leq n-w+1$. Moreover, since the projection, on the nonzero entries,
of the codewords whose support is a given subset of $w$ coordinates forms an OA$(w-d+1,w,q-1)$, it follows
that $\cS$ contains a codeword with a nonzero symbol in each one of the last $n-w+1$ coordinates and hence $\abs{S} \geq n-w+1$.

Thus, $\abs{S} = n-w+1$.

\item Each two codewords of $\cS$ differ in all the symbols of their first $d-2$ coordinates.
This implies that $q-1 \geq \abs{S}$.
\end{enumerate}

Thus, $q-1 \geq n-w+1$, which completes the proof of the theorem.
\end{proof}

\subsection{Comparison between maximum size anticodes}
\label{sec:anticodes}

So far we have characterized the families of diameter perfect constant-weight codes.
Each family is associated with some maximum size anticodes. In this subsection
we will characterize these families of maximum size anticodes and compare some of them.

The first family of maximum size anticodes is associated with the family [{\bf [F1]}] of non-binary diameter perfect constant-weight codes
for which $w=n$. Clearly, for these anticodes the length of a codeword is $n$ and the weight of each codeword is $w=n$.
Moreover, the anticodes are derived from the related anticodes in the Hamming scheme, by replacing the \emph{zeroes} in the
anticodes of the Hamming scheme with the additional nonzero symbol of the constant-weight code.

The second family of maximum size anticodes is associated with the family
[{\bf [F2]}] of diameter perfect constant-weight codes over an alphabet of size $2^k+1$ for which $w=n-1$.
Clearly, the related anticodes also have length $n$ and weight $w=n-1$.
If the non-binary diameter perfect code is in fact a non-binary perfect code, then the related anticode is a ball.
If the non-binary diameter perfect code is not a non-binary perfect code, then the related anticode is not a ball
and it should be computed for each set of parameters. For example, it was proved in~\cite{KOP16} that for ternary codes if $n=2^m \geq 8$,
$w=n-1$, an the diameter is 4, then the maximum size anticode has size $n^2$ and such a set can be defined by
the union of the set of ternary words with a unique \emph{zero} and all the other symbols are \emph{ones} and
the set of ternary words with a unique \emph{zero} and a unique \emph{two} and all the other symbols are \emph{ones}.

The third family of maximum size anticodes is associated with the family of generalized Steiner system.
The related anticode $\cA^\text{s} (n,w,t)$ was defined by
$$
\cA^\text{s} (n,w,t) \triangleq
\{ (\overbrace{1\cdots \cdots 1}^{t ~ \text{times}}, a_1,\ldots,a_{n-t} ) ~:~ a_i \in \Z_q, ~ \wt{a_1 \cdots a_{n-t}} = w-t \}~.
$$
This anticode has diameter $2(w-t)$ (when $n-t \geq 2(w-t)$) and its size is $\binom{n-t}{w-t} (q-1)^{w-t}$.

The last family of maximum size anticodes is associated with families [{\bf F4}], [{\bf F5}], and~[{\bf F6}] of
the diameter perfect constant-weight codes. The related anticode $\cA^\text{m} (n,w,\delta)$ was defined by
$$
\cA^\text{m} (n,w,\delta) \triangleq \{ (a_1, \ldots, a_\delta , \overbrace{1\cdots \cdots 1}^{w-\delta ~ \text{times}} ,
\overbrace{0\cdots \cdots 0}^{n-w ~ \text{times}}) ~:~ a_i \in \Z_q \setminus \{ 0 \} \}~,
$$
This anticode has diameter $\delta$ and its size is $(q-1)^\delta$.

One can be easily observed that nontrivial anticodes of the first two families cannot have the same parameters since they have different weights.
Moreover, it can be observed that nontrivial anticodes from these two families cannot have the same parameters
as the anticodes from the last two families. Hence, we will compare the anticodes from the last two families.

Unless the two anticodes represent one of two trivial cases ($w=n$ and $\delta=w-t$; or $w=t$) they
cannot be isomorphic. This can be observed from the fact that the \emph{zeroes} of $\cA^\text{m} (n,w,\delta)$
are in $n-w$ fixed coordinates, while the \emph{zeroes} of $\cA^\text{s} (n,w,t)$ are in any combination
of $n-w$ coordinates in the last $n-t$ coordinates.

Can these two anticodes be maximum size anticodes, related to two diameter perfect codes of different families,
be of the same size (when the length, weight, and diameter are the same)? Note first that this implies that the
related code from the family [{\bf F4}], or the family [{\bf F5}], or the family [{\bf F6}] must be also a generalized Steiner system
since the two codes will have the same parameters.
Since $\abs{\cA^\text{s} (n,w,t)} = \binom{n-t}{w-t} (q-1)^{w-t}$ and $\abs{\cA^\text{m} (n,w,\delta)} = (q-1)^\delta$, it follows that
the two anticodes are of equal size if and only if $\binom{n-t}{w-t} = (q-1)^\ell$ for some nonnegative integer~$\ell=\delta+t-w$.
If $\ell=0$, then either $n=w$ (the first trivial case) or $w=t$ (the second trivial case). We distinguish now
between two cases depending on whether $\ell =1$ or $\ell > 1$.
\begin{enumerate}
\item If $\ell=1$, then $\binom{n-t}{w-t} = q-1$ and we distinguish between three cases, depending on whether
the related code is from the family [{\bf F4}], the family [{\bf F5}], or the family [{\bf F6}].

\noindent
{\bf Case 1.1.} The diameter perfect code is from the family [{\bf F4}] which implies that $\delta =w$ and hence $t=\ell=1$.

Since $t=1$, it follows that $\binom{n-1}{w-1} =q-1$ and hence by Corollary~\ref{cor:F4trivial},
there exists an $(n,w+1,w)_q$ code of the family [{\bf F4}].
If $w=2$, then $q=n$ and the two related codes are a GS$(1,2,n,n)$ which is also an $(n,3,2)_n$ code from the family [{\bf F4}].
Such a code exists by Corollary~\ref{cor:w=2F4}.
If $2 < w < n$, then a related code with $\binom{n}{w}$ codewords cannot be a GS$(1,w,n,q)$.

\noindent
{\bf Case 1.2.} The diameter perfect code is from the family [{\bf F5}] which implies that $\delta =w-1$ and hence $t=\ell+1=2$.

For the GS$(2,w,n,q)$ and the $(n,w,q)$ MDS-CW code to be equal they must have the same minimum distance and hence $2(w-2)+1=w$, i.e., $w=3$.
Since also $\binom{n-2}{w-2} = q-1$, it follows that $n=q+1$. Two codes are considered in this case.
The first one is a generalized Steiner system GS$(2,3,q+1,q)$ derived from a 1-perfect Hamming code over $\F_q$.
The second one is an $(q+1,3,q)$ MDS-CW code derived from a $[q+1,q-1,3]_q$ MDS code. For these parameters the 1-perfect
Hamming code is also an MDS code and hence both constant-weight codes are the same code. There might be other such constant-weight codes
for $q$ which is not a power of a prime, but no such code is known.

\noindent
{\bf Case 1.3.} The diameter perfect code is from the family [{\bf F6}], for which the MOA-CW code has minimum
Hamming distance $d < w$, which implies that $\delta =d-1$ and hence $t=\ell+1 +w-d= w-d +2$.

Hence, the related codes are GS$(t,w,n,q)$ and an $(n,d,w)_q$ MOA-CW code.
The codes have the same minimum Hamming distance and hence $d=2(w-t)+1=2d-3$, i.e., $d=3$, which
implies that $w=t+1$. Since $\binom{n-t}{w-t} = q-1$, it follows that $n-t=q-1$, i.e., $n=q+t-1$, and hence one of our codes is
a GS$(t,t+1,q+t-1,q)$. By iteratively applying Lemma~\ref{lem:derivedGS}, we obtain a GS$(2,3,q+1,q)$ which is the code in the previous case.
Unfortunately, no GS$(t,t+1,q+t-1,q)$ is known for $t>2$.

\item If $\ell>1$, then $\binom{n-t}{w-t} = (q-1)^\ell$ and first we have to consider the solutions for this equation.
We distinguish between three cases depending whether $w-t \in \{1,n-t-1\}$, $w-t \in \{2,3,n-t-3,n-t-2\}$, or $3 < w-t < n-t-3$.

\noindent
{\bf Case 2.1.} If $w-t \in \{1,n-t-1\}$.

If $w-t=n-t-1$, then $w=n-1$ and one code is a GS$(t,n-1,n,q)$ and by iteratively applying Lemma~\ref{lem:derivedGS}, we obtain
a GS$(1,n-t,n-t+1,q)$. One can easily verify that for such a code $n-t+1 \geq 1 + (n-t+1)(q-1)$ (see also~\cite{Etz97}) and hence
it does not exist.

If $w-t=1$, then one code is a GS$(t,t+1,n,q)$ for which the minimum distance is 3. Hence, the related codes from the
families [{\bf F4}], [{\bf F5}], and [{\bf F6}] are only those considered in cases 1.1., 1.2., and 1.2., respectively.
Therefore, no such code will be found for $\ell >1$.

\noindent
{\bf Case 2.2.} If $w-t \in \{2,n-t-2\}$.

When $w-t=2$ or $w-t=n-t-2$, there are infinitely many such solutions which satisfy
the recursion $a_m = 6a_{m-1} - a_{m-2}$ (where $q-1=a_m$ and $\ell=2$)
[Online Encyclopedia on Integer Sequences A001109], with the initial conditions $a_1=1$ and $a_2=6$.
Similar analysis to the previous cases shows that there is no diameter perfect code from two different families in this case.

\noindent
{\bf Case 2.3.} If $2 < w-t < n-t-2$.

It was proved by Erd\"{o}s~\cite[p. 48]{Lio83} that this equation
has exactly one solution for $n-t=50$ and $w-t=3$ or $w-t=n-t-3$. In the region for this solution there is no code from two families.
\end{enumerate}

Therefore, all those cases for which the anticodes $\cA^\text{s} (n,w,t)$ and $\cA^\text{m} (n,w,\delta)$ have the same size and different structure, the anticodes are not related to a diameter perfect codes from two different families. Moreover, they might not be of maximum size.

In general, one can decide based on the size of a maximum size anticode if the given parameters are in the range of a generalized Steiner system,
an ${\text{MDS-CW}}$ code, or an ${\text{MOA-CW}}$ code. Each such code is an optimal non-binary constant-weight code that meets the value
of $A_q(n,d,w)$. In some cases these codes coincide as illustrated in the following example (and analyzed in the comparison
of the codes with their anticodes).

\begin{example}
Let $\cC$ be a linear 1-perfect code of length $q+1$, dimension $q-1$, and minimum Hamming distance 3, over $\F_q$.
By its parameters, this code is also an MDS code.
The codewords of weight three of $\cC$ form a generalized Steiner system GS$(2,3,q+1,q)$ and also
a $(q+1,3,q)$ MDS-CW code. The related maximum size anticodes are $\cA^\text{s} (q+1,3,2)$ and $\cA^\text{m} (q+1,3,1)$
which are of the same size $(q-1)^2$, but different structure.
\end{example}

\section{Conclusions and Open Problems for Future Research}
\label{sec:conclude}

We have considered diameter perfect constant-weight codes.
First, we have revisited the family of such binary constant-weight codes which are codes in the Johnson scheme.
Non-binary such codes (where the metric is the Hamming distance) are not associated with an association scheme and hence the
original proofs for the code-anticode bound do not hold for these codes. Also proofs which were given to other spaces
with related metrics, which do not require the association scheme
conditions, do not hold for these codes. We have presented a novel proof to the code-anticode bound for these codes.
We have distinguished between six families of such codes and four families of related maximum size anticodes.
All the new constructed non-binary diameter perfect constant-weight codes are optimal codes, which attain the value of $A_q(n,d,w)$.
Two of the families of anticodes are new and the proof of their optimality is a simple consequence from the code-anticode bound.
As was pointed out in~\cite{AAK01} maximum size anticodes in the Hamming scheme, the Johnson scheme, and the Grassmann scheme, are related
to $t$-intersecting families and the celebrated Erd\"{o}s-Ko-Rado theorem. The proofs  that certain $t$-intersecting families (anticodes)
are of maximum size are not simple and they were mainly obtained via extremal combinatorics.
This is in contrast to our proofs which are very simple since they are derived from the code-anticode bound.
Four of the families of diameter perfect constant-weight
codes were considered before, but they were not observed as diameter perfect constant-weight codes. The last two families
are new and raise many problems for future research. Some of these problems, as well as other problems related
to the other families, will be presented now.

\begin{enumerate}
\item Are there diameter perfect codes in J$(n,w)$, except for Steiner systems and their complements?
By Conjecture~\ref{conj:Dbinary} such codes do not exist, but there are hardly any result in this direction. It was proved in~\cite{AAK01}
(who used an idea from~\cite{Etz96}) that if such a diameter perfect code exists then some Steiner systems also exist. This excludes
the existence of some parameters of such diameter perfect codes as a consequence of the necessary conditions in Corollary~\ref{cor:necesary_ST}.

\item Continue to develop the theory of generalized Steiner systems GS$(t,w,n,q)$. Even for small parameters,
such as $(t,w)=(2,3)$, we don't know if the necessary conditions similar to the ones in Corollary~\ref{cor:necesary_ST} are sufficient,
although there are many results on these parameters, e.g.~\cite{CGZ99,Etz97,PhYi97}. We conjecture
that this is the case for all $q$, with a possible exceptions for a small number of values of $n$.
For other parameters there are also lot of research work, e.g.~\cite{CJZ07,Etz97,WuFa09}, but there is no known
construction for $t > 3$ and we would like to see a construction for such a system.

\item For a given pair $(w,n)$, $3 \leq w < n$, find good upper bound on the alphabet size QMDS$(w,n)$, for which
there exist an $(n,w,q)$ MDS-CW code for each $q \geq \text{QMDS}(w,n)$. It was proved in~\cite{Etz97} that for each pair $(w,n)$
there exists such a value QMDS$(w,n)$, but the current upper bound is very large.

\item For any given pair $(w,n)$, $4 \leq w \leq n-2$, improve the upper bound on $q=q_0 (w,n)$, which is the smallest alphabet
for which there exist an $(n,w+1,w)_q$ code. A very intriguing problem is this direction is to find good lower bounds
and upper bounds on $\chi (w,n)$, the minimum number of codes in a partition of all binary words of length $n$ and weight $w$
into codes with minimum Hamming distance $w+2$.

\item Present new constructions for MDS-CW codes (and also for MOA-CW codes). The only construction, which is not derived directly
from an orthogonal array, which we gave was analyzed in Theorem~\ref{thm:main_MOA-CW}. Another construction is the union
construction for ${\text{MDS-CW}}$ codes as mentioned in Theorem~\ref{thm:const_union}. Is there a related construction for ${\text{MOA-CW}}$ codes?
We would like to see new different constructions as well as amendments to the construction which was given in the paper.

\item Present new bounds on the tradeoff between the parameters of MDS-CW codes (and also for MOA-CW codes).
This is especially important and we would like to see some new directions which will enable to conjecture for which
parameters such codes exist.

\item Given $1 < d < w < n$, does there exist a $q_0 (n,d,w)$ for which there exists an $(n,d,w)_q$ MOA-CW for all $q \geq q_0 (n,d,w)$?
recall that for MDS-CW codes such a value called QMDS$(w,n)$ exists as was proved in~\cite{Etz97}.

\item Are there more families of diameter perfect constant-weight codes, except for the six families which were presented?
One possible direction, to exclude such possible families of codes, is to show sets of parameters with tradeoff between $n$, $w$, and $d$,
where such codes cannot exist.

\item Characterize all parameters for which $\cA^\text{s} (n,w,t)$ is a maximum size anticode. Such a proof can be done
by using extremal combinatorics as was done in~\cite{AhKh97} for such anticodes related to binary words with constant weight.
Similarly, characterize all parameters for which $\cA^\text{m} (n,w,\delta)$ is a maximum size anticode. Such a proof can be also done
by using extremal combinatorics as was done in~\cite{AhKh97} and in~\cite{AhKh98} for such anticodes related to binary words
with constant weight and for non-binary words in the Hamming scheme. The range in which $\cA^\text{m} (n,w,\delta)$ is optimal
can be proved by using a combination of techniques from both papers.

\item Are there $(q+1,3,q)$ MDS-CW codes which are also generalized Steiner systems GS$(2,3,q+1,q)$ beside those for prime power $q$?

\item Is there any GS$(3,4,q+2,q)$ for a prime power $q$? and for non-prime  power $q$?

\item Define the anticodes $\cA^\text{s} (n,w,t)$ and $\cA^\text{m} (n,w,\delta)$ in terms of $t$-intersecting families.
Find the maximum size of these $t$-intersecting families (in other words, the maximum size of these anticodes) for all
parameters, including the cases where $d > w+1$.
\end{enumerate}

\section*{Acknowledgement}

The author would like to thank Denis Krotov for bringing~\cite{BDMW,Kro08,KOP16} to his attention.
He also would like to thank Daniella Bar-Lev for bringing~\cite{Lio83} to his attention.


%


\end{document}